\documentclass[a4paper,10pt]{article}
\usepackage[utf8]{inputenc}
\usepackage{amssymb}
\usepackage{amsmath}
\usepackage{amsthm}
\usepackage{color}
\usepackage{graphicx}
\usepackage{epstopdf}
\usepackage{xcolor}
\usepackage{mathtools}
\usepackage{authblk}
\usepackage{a4wide}

\providecommand{\keywords}[1]{\textbf{\textit{Index terms---}} #1}
\usepackage{array}
\usepackage{verbatim}

\newtheorem{claim}{Claim}[section]
\newtheorem{theorem}{Theorem}[section]
\newtheorem{lemma}{Lemma}[section]
\theoremstyle{definition}
\newtheorem*{remark}{Remark}

\title{Species tree estimation using ASTRAL: how many genes are enough?}
\author[1]{Shubhanshu Shekhar}
\author[2]{Sebastien Roch}
\author[1]{Siavash Mirarab\thanks{smirarab@ucsd.edu}}

\affil[1]{Department of Electrical and Computer Engineering, University of California, San Diego}
\affil[2]{Department of Mathematics, University of Wisconsin-Madison, Wisconsin}

\begin{document}

\maketitle

\let\thefootnote\relax\footnotetext{Accepted for oral presentation at RECOMB 2017\\ Under review at IEEE TCBB (submitted on 12/26/2016)}

\begin{abstract}
Species tree reconstruction from genomic data is increasingly
performed using methods that account for sources of gene tree discordance such
as incomplete lineage sorting. One popular method
for reconstructing species trees from unrooted gene tree topologies is
ASTRAL. In this paper, we derive theoretical sample complexity results
for the number of genes required by ASTRAL to guarantee
reconstruction of the correct species tree with high probability. We also validate those theoretical bounds
in a simulation study. Our results indicate that ASTRAL requires
$O(f^{-2}\log n)$ gene trees to reconstruct the species tree correctly with high probability
where $n$ is the number of species and $f$ is the length of the shortest branch 
in the species tree. 
Our simulations, some under the anomaly zone, 
show trends consistent with the theoretical bounds and also provide some practical
insights on the conditions where ASTRAL works well. 
\end{abstract}


\section{Introduction}
\label{sec:introduction}

Evolutionary
relationships between organisms are typically represented
in the form of a phylogeny known as a species tree.
Each branching point in a species tree represents
a historical speciation event. 
On the other hand, phylogenies of individual
parts of genomes of extant species do not always follow the same patterns
as the speciation events. 
This discordance can be
due to various biological processes that include
incomplete lineage sorting (ILS), 
duplication and loss, horizontal gene transfer, 
and hybridization~\cite{Maddison1997,Degnan2009,Bapteste2013}.
These processes can impact not only functional genes but any region of the genome; however, following the convention, we use the term {\em gene tree} to refer to a tree specific to a single region of the genome. 
The potential discordance between gene trees and species tree has motivated a growing set of methods for inferring species trees while accounting for these differences~\cite{Edwards2007,mpest,astral,astral2,njst,Heled2010,SVDquartets,star,glass,metal,Chaudhary2013,Wehe2008}.
Most existing techniques
focus on one of the biological causes
of discordance (but there are exceptions~\cite{Szollsi2014,Boussau2013,Rasmussen2012,Yu2014}). In particular ILS, a ubiquitous cause of discordance found in many datasets~\cite{1kp-pilot,Pollard2006,Song2012},
has been the subject of much interest.  

ILS is best understood by considering the interplay between phylogenetics
and the population genetic coalescent process~\cite{Kingman1982}.
The most widely used model to study ILS is the multispecies coalescent model (MSC)~\cite{Pamilo1988,Rannala2003}.
Under the MSC model, to generate a gene tree, the Kingman's coalescent
process~\cite{Kingman1982} is followed on each species tree branch,
which represents a population and is often assumed to have a fixed population size; all lineages that reach the top of a branch are copied to the ancestral population, which follows its own coalescent process. 
When two lineages from two different species fail to coalesce in their
lowest common ancestral population, they both go back to an earlier branch where alleles from
other species are also present; there, the two lineages may first coalesce with lineages from those other species before coalescing with each other.
When this happens, the resulting lineage tree differs from the species tree, and is said to experience ILS.

%In the face of ILS, a gene can maintain multiple alleles throughout the lifespan of an ancestral population and those alleles get fixated (or sampled) in descendant species in ways that do not match the speciation events. Thus, 

When only concerned with gene tree topologies, species tree branch lengths can be specified in coalescent units (the number of generations divided by the effective population size)~\cite{Degnan2009}. 
A species tree with branch lengths in coalescent units uniquely defines a distribution on gene tree topologies~\cite{Degnan2005} under
the MSC model.

Species tree estimation despite ILS is possible using an array of methods that use various strategies, including co-estimation of gene trees and species tree~\cite{best,Heled2010,Larget2010},
species tree estimation without gene trees~\cite{SVDquartets,snapp}, and two-step approaches, which first estimate all gene trees independently and then combine them using a summary method (e.g., \cite{glass,metal,star,mpest}). Since gene tree inference typically produces unrooted trees, summary methods that operate on unrooted trees
are more useful.
An early unrooted method was NJst~\cite{njst}, 
which can be considered~\cite{Allman2016} the unrooted version of 
STAR, and is implemented efficiently in a tool
called ASTRID~\cite{astrid}.
A newer suite of methods called ASTRAL~\cite{astral,astral2}
has also been designed and has been widely used on biological datasets (e.g., ~\cite{Yang2015,Laumer2015,Giarla2015,Rothfels2015,Andrade2015,1kp-pilot,Cannon2016,Rouse2016,Dutheil2016,Yuan2016,Huang2016}).

The ASTRAL suite of tools seek to solve the Maximum Quartet Support Species Tree (MQSST) problem. 
Each quartet of leaves can have one of three unrooted tree topologies and, for each quartet
selected from the full set of leaves, the species tree will induce one of the three topologies. 
The quartet score of a species tree with respect to a set of gene trees is the sum of the number of induced quartet topologies shared between the species tree and each gene tree (see Appendix~\ref{subsec1} for a formal definition).
The MQSST problem is to find the species tree with the maximum quartet score with respect
to a set of input gene trees. 
The problem is NP-hard~\cite{Lafond2016}, and this has led to 
the development of several versions of ASTRAL (Table~\ref{tab:sum}). 
The exact version of ASTRAL, which we call ASTRAL*, uses dynamic programming
to solve the problem exactly in exponential time.
ASTRAL* has limited scalability, motivating the definition of 
a constrained version of the MQSST problem
that restricts the species tree to draw its branches from a predefined set of bipartitions. 
In ASTRAL-I, the constraint set is the set of bipartitions in the input gene trees~\cite{astral}. ASTRAL-II further expands the 
constraint set heuristically using several complicated and 
empirically motivated strategies~\cite{astral2}
that evade straightforward mathematical formulation. 

ASTRAL*, ASTRAL-I, and ASTRAL-II are all statistically consistent estimators of the species tree given true gene trees~\cite{astral}.
The proof of consistency follows from a result by
Allman {\em et al.} that shows that, for any quartet of leaves, the species tree topology has no less than a 1/3 probability of matching each gene tree~\cite{Allman}; thus, the most likely quartet gene tree matches the species tree. 
However, this result does not extend to more than four species.
In fact, for five species or more, there always exist species trees where the probability of an unrooted gene tree matching the unrooted species tree is less than some other tree topology~\cite{Degnan2013}.
A species tree that does not match the most likely gene tree is said to be in the ``anomaly zone''~\cite{Degnan2006,Degnan2013,Rosenberg2013}.

\begin{table}[!t]
\renewcommand{\arraystretch}{1.3}
\caption{Versions of ASTRAL, their
constraint set ($X$), and time complexity with 
$n$ (the number
of species), $m$ (the number of genes), and $|X|$.}
\label{tab:sum}
\centering
% Some packages, such as MDW tools, offer better commands for making tables
% than the plain LaTeX2e tabular which is used here.
\begin{small}
\begin{tabular}{|c||c|c|}
\hline
\textbf{Version} & \textbf{Constrain set ($X$)} & \textbf{Time}\\
\hline
ASTRAL* & Unconstrained (i.e., exact)& $O(nm3^n)$\\
\hline
ASTRAL-I& Gene tree bipartitions& $O(n^4m^3)$\\\hline
ASTRAL-II& ASTRAL-I$+$heuristic additions&$O(nm|X|^{2})$ \\
\hline
\end{tabular}
\end{small}
\end{table}

Despite widespread use of the ASTRAL suite and its high accuracy in simulations~\cite{astral,astral2,astrid,distique}, little is understood about 
its sample complexity. In this context, the sample complexity is simply the number of genes ($m$) required to guarantee correct species tree reconstruction with high probability.
Sample complexity is typically established asymptotically and with respect to the number
of species ($n$) or the length of the shortest branch in the species tree ($f$).

Under the MSC model, shorter branches in coalescent units result in increased gene tree discordance and consequently increase the difficulty of reconstructing the species tree.
Sample complexity results have been established
for some methods that are not in wide use. 
For instance, for sufficiently small $f$, the GLASS and STEAC algorithms, which use both gene tree branch length and topology, need the number of true gene trees to scale linearly with $f^{-1}$ and $f^{-2}$, respectively~\cite{Roch2013}. 
For summary methods, no better bound 
has yet been found. More specifically, 
for algorithms that only use gene tree topology,
the best demonstrated data requirement is $f^{-2}$.
It is worth noting that these results are all 
assuming that true gene trees are known. We revisit
the issue
of gene tree error and review known results for methods 
that do not rely
on input gene trees in the discussion section.

The only previous sample complexity results specifically relevant to
ASTRAL-II is the bipartition cover results established by Uricchio {\em et al.}~\cite{rosenberg2016}.
They establish
loose upper bounds on the required number of genes such that each bipartition
in the species tree is observed in at least one of the gene trees with high probability.
While this result does not help in analyzing the data requirement of
ASTRAL*, it can be used to limit the complexity of ASTRAL-I and ASTRAL-II once the complexity of ASTRAL* is established.

In this paper, we establish the first theoretical bounds for the sample complexity of 
ASTRAL*.
We show that the sufficient number of true gene trees required by ASTRAL* to reconstruct the correct species tree with high probability grows quadratically with the inverse of the shortest branch length in the species tree and logarithmically with the number of species. We further show that the necessary number of genes for ASTRAL* to reconstruct the species tree has a similar asymptotic behavior.
We then test the data requirements of ASTRAL-II in simulations and show that its empirical data requirements match the theoretical results for ASTRAL*. 
Our simulation study focuses on datasets with extremely short
branches, including some conditions that fall under the anomaly zone, 
and the results indicate that
ASTRAL-II can be made arbitrarily accurate given enough 
true gene trees even under very difficult conditions.  
We also test the data requirement of NJst~\cite{njst}, as implemented in the ASTRID~\cite{astrid} software package, in simulations and show that the relative performance of ASTRAL-II and NJst depends on the shape of the species tree.  

\section{Theoretical bounds}

We first present theoretical bounds on the number of gene trees required by the exact ASTRAL* algorithm to reconstruct the true species tree with high probability. We present the main results
here and leave the proofs for the appendix. 
%We  show that the data requirement of ASTRAL increases quadratically with the inverse of the length of the shortest branch in
%the species tree. 

Let $n$ be the number of leaves in the species tree, and let $m$ denote the number of gene trees, all generated
from the species tree according to the MSC model.
We use $f$ to denote the length of the shortest species tree branch
(as an unrooted tree).  
For each quartet of leaves, there are three possible unrooted tree topologies that can be induced by  each gene tree. 
Under the MSC model, we know that for each quartet $q_i$, the probability that a gene tree will be congruent with the species tree is $p_i = 1 -\frac{2}{3}e^{-d_i}$, and the other two configurations have equal probability ($r_i = (1-p_i)/2$) of occurring; here,
$d_i$ gives the length of the path between the two endpoints
of the middle edge of the quartet topology in the species tree, measured in coalescent units~\cite{Allman}. 
ASTRAL* returns a fully binary species tree which shares the maximum number of induced quartets with the input gene trees~\cite{astral} (Appendix~\ref{subsec1}). 
We say ASTRAL* has an error 
when there exists a bipartition in the reconstructed tree that does not appear in the true species tree.

\subsection{Upper bounds}
We first ask how many gene trees are sufficient
for ASTRAL* to reconstruct the true species tree with high probability. 
%Our main results is the following theorem. 

\begin{theorem}\label{thm:upper}
Consider a model species tree with minimum branch length $f< \log(\sqrt{2})$. Then, for any $\epsilon>0$, ASTRAL* returns the true species tree
 with probability at least $1-\epsilon$ if the number of input error-free gene trees satisfies 
 \begin{equation}
 \label{eq:upper_bound}
  m > \frac{9}{2}\log\bigg(\frac{4{n \choose 4}}{\epsilon}\bigg)\frac{1}{(1-e^{-f})^2}.
 \end{equation}
\end{theorem}

\begin{remark}
In the limit of small $f$, we can use the fact that $\lim_{f\downarrow 0}\frac{1-e^{-f}}{f} \uparrow 1$, and thus for $f$ small enough we have $\frac{1-e^{-f}}{f} \geq \sqrt{0.9} $, and so a sufficient condition on the number of gene trees required is
\begin{equation}
\label{eq:upper_bound2}
m > 20\log\bigg(\frac{n}{6\epsilon}\bigg)\frac{1}{f^2}.
\end{equation}
That is, for small values of $f$,   $m=\Omega(f^{-2} \log n)$ gene trees
 are sufficient for ASTRAL* to return the correct
 species tree with an arbitrarily high probability. 
\end{remark}
Our proofs, shown in Appendix~\ref{sec:p-upper}, are based on the observation that quartet scores follow a 
multinomial distribution
and therefore,
for a large number of genes, we  expect the frequencies 
to
concentrate 
tightly around their means. 
We use 
 Hoeffding's inequality to get the required concentration.

%\TODO{depending on the new form of the theorem, 
%rewrite or remove. }

%COMMENTING OUT THE PREVIOUS REMARK

\begin{comment}
\begin{remark}
Theorem~\ref{thm:upper} shows us that when $f$ is small enough, $\mathcal{O}(f^{-2}\log(n))$ genes are sufficient for reconstructing the species tree with high probability. 
As an example, consider $\alpha_0 = \log(2/\sqrt{3})$, which gives us that
for $f < 0.14 <\log(2/\sqrt{3})$, using $m > 9\log\big(\frac{4}{\epsilon}{n \choose 4}\big)\frac{1}{f^2}$ as input to ASTRAL* results in the correct species tree with probability at least $1-\epsilon$. 
\end{remark}

\end{comment}

\subsection{Lower bounds}

%We now turn to bounds on the necessary 
%number of gene trees for ASTRAL* to return the true species tree. 
We now show that there exist  species trees for which $\Omega(f^{-2}\log n)$  gene trees are also necessary. 

\subsubsection{Quartet species trees ($n=4$)}
It is simpler to first analyze the case of $n=4$, i.e., reconstructing a single quartet.
Let $\mathcal{Y} = \{q_1,q_2,q_3\}$ be the set of three possible 
topologies, let $\{n_{11},n_{12},n_{13}\}$ be their frequencies out of $m$ gene trees and let $Q = \{p,r,r\}$ be the corresponding probability distribution on $\mathcal{Y}$ according to the MSC model. Let $E$ denote the event that ASTRAL* returns the wrong species tree. Observe that,
for $n=4$, the error event of ASTRAL* is equivalent to $n_{11} < \max\{n_{12},n_{13}\}$. 
Let $F$ denote the event $\{ n_{11} < m/3\}$, which requires $\max(n_{12},n_{13})>m/3$. ASTRAL* returns the tree associated with the highest frequency; thus, $P(E) \geq P(F)$.

 \begin{figure}[tbp]
 \includegraphics[width=.471\textwidth]{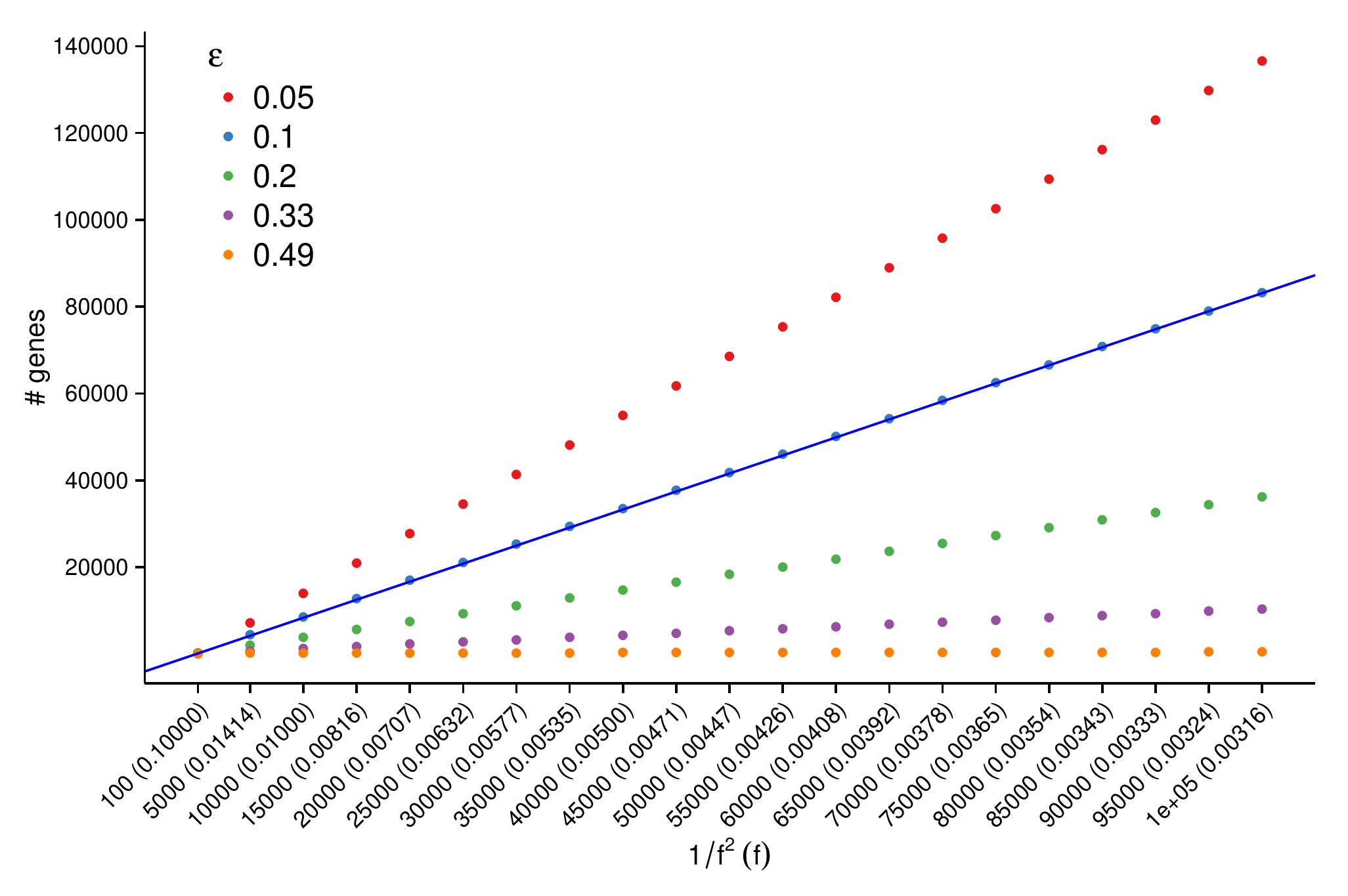}
 \caption{The number of genes below which a certain level of 
 error ($\epsilon$) is unavoidable according to the Binomial distribution. 
Values on the x axis show $1/f^2$ where $f$ is the branch length
 in coalescent units
 ($f$ is shown parenthetically). The y axis shows
 the largest value of $m$ where the CDF of a binomial
 with parameter $1-{2}/{3}e^{-f}$ evaluated
 at ${m/}{3}$ is above the 
 error level $\epsilon$. A line is fitted to points
 for each $\epsilon$ (shown for each line).\label{fig:binom}}
\end{figure}

Before proving the lower bounds, we first present a simple numerical experiment. 
Suppose we want to find the $m$ below which the error probability is at least $\epsilon$ for any $\epsilon \in [0,0.5)$. 
We know that $P(F)$ is a lower bound on $P(E)$, and that it has the following closed form expression
\begin{equation}
P(F) = \sum_{i=0}^{\lfloor m/3 \rfloor} {m \choose i}p^i(1-p)^{m-i},
\end{equation}
which is simply the Cumulative Distribution Function (CDF) of a binomial calculated at $\lfloor m/3 \rfloor$, where $p= 1 -\frac{2}{3}e^{-f}$. 
For given $f$ and $\epsilon$, we 
can
compute the value of $P(F)$ for various choices of $m$. 
In  Figure~\ref{fig:binom}, we show the largest $m$
for  which the computed $P(F)$ is larger than $\epsilon$.  
As $f^{-2}$ increases,
the required number of binomial trials 
grows  
almost linearly with $f^{-2}$, with a slope that 
increases with decreasing $\epsilon$,
as expected. 
This experiment suggests that for $m=\mathcal{O}(f^{-2})$, the reconstruction error for one quartet can be made arbitrarily close to 
$1/2$. Our next result formalizes this intuition. 

\begin{theorem}
\label{one_q}
For a species  tree with $n=4$ leaves, let $f$  be the length of the single internal branch in the unrooted species tree.
 Then, for any $\epsilon \in [0,0.5)$, there exists an $f_0 >0$ and  a constant $d_{\epsilon}$,  such that for all $m \leq d_{\epsilon}/f^2$ with $f \leq f_0$, the probability that ASTRAL* outputs the wrong quartet is at least $\epsilon$. \end{theorem}

This theorem shows that with only four species and $m = \mathcal{O}(f^{-2})$ genes, the 
probability of error for ASTRAL* can be made arbitrarily close to 1/2. 
Thus, the lower bound and the upper
bound asymptotically match as $f$ decreases. 

The proof, shown in Appendix~\ref{sec:p-one_q},
essentially uses the Berry-Ess\'{e}en theorem
to first approximate the binomial with a 
Gaussian distribution and then shows
that for sufficiently small $f$,
 when the number of observations
is below $\mathcal{O}(f^{-2})$, the 
Gaussian has a sufficiently large probability 
of deviating from its mean enough 
to become below 1/3.

\subsubsection{Species trees of arbitrary size ($n>4$)}

If all quartets in a gene tree were generated independently, extending Theorem~~\ref{one_q} to more species would be trivial. However quartets have a complicated dependence structure. 
To establish a complement to our upper bound, 
it suffices to show that there exists a tree where $m=\mathcal{O}(f^{-2}\log n)$
gene trees would result in a high probability of error. 
We start with a simpler claim that establishes the lower bound with respect
to 
$f$---but not $n$.

\begin{claim}
\label{simple_1}
For any $n$ and $\epsilon \in [0,0.5)$, there exists a species tree with $n$ leaves and shortest
branch length $f$ such that when ASTRAL* is used with  $m \leq d'_{\epsilon}/f^2$ gene trees, for some constant $d'_{\epsilon}$, the probability that ASTRAL* reconstructs the wrong tree
is at least $\epsilon$. 
\end{claim}

The proof, detailed in Appendix~\ref{sec:p-simple_1},
simply uses trees of the form $(((u,v),w),\mathcal{X})$, where $\mathcal{X}$ is an arbitrary tree with $(n-3)$ leaves (represented by the set $X = [n-3]$) that is connected
to the root with a very long branch. 
By making this branch sufficiently long, 
we ensure that with high probability 
all quartets with leaves $\{x,u,v,w\}$ for $x\in X$
have the same frequencies in the gene 
trees; this observation enables us
to reduce
the problem to the case of $n=4$ and make use of
 Theorem~~\ref{one_q}.

%\subsubsection{Lower bounds matching upper bounds}

By construction, our result in Claim~\ref{simple_1} does not
depend on $n$. To match the lower bound, we strengthen this result to an upper bound that depends on
both $f$ and $n$. 

\begin{theorem}\label{thm:upperfull}
For any $\rho \in (0,1)$ and $a \in (0,1)$, there exist constants $f_0$ and $n_0$ such that the
following holds. For all $n \geq n_0$ and $f \leq f_0$,
there exists a species tree with $n$ leaves  and shortest
branch length $f$ such that when ASTRAL* is used with $m \leq \frac{a}{5}\frac{\log n }{f^2}$ gene trees, the event $E$ that ASTRAL* reconstructs the wrong tree has probability
\begin{equation}
\label{eq:new_lower}
P(E) \geq 1-\rho. 
\end{equation}
\end{theorem}

%\TODO{Check and improve the theorem statement}

\begin{remark}
The above result tells us that, for all $n$ large enough, there exists a species tree such that when ASTRAL* is used with $m \leq \frac{a\log n }{5f^2}$ gene trees, the probability of error is bounded away from zero. This lower bound can be made arbitrarily close to 1 (as $n$ goes to infinity). 
%However, for most practically relevant values of $n$, this lower bound can be quite a small number. For example, with $a=3/8$ and $n=256$ we get that $P^m(E) \geq (1-\rho)\times 0.043$ according to the RHS in Eq.~\ref{eq:new_lower}. 
\end{remark}
The proof (Appendix~\ref{sec:p-upperfull}) uses ideas
similar to Claim~\ref{simple_1}; 
instead of using one  construct of the form $((u,v),w)$ matched
with $x\in \mathcal{X}$ (as we did for Claim~\ref{simple_1}), 
we use $n/3$ such constructs (each called a triplet). 
We show that each triplet, when matched
with one more leaf to build a quartet, 
has a probability of error of at least $\Omega(\frac{1}{n^a})$. We then argue that 
when branches above all triplets 
are long enough, error events for individual triplets can be made independent
conditioned on each triplet coalescing
in its root branch.
Finally, the conditional independence of $\Omega(n)$ error events and the 
$\Omega(\frac{1}{n^a})$ lower bounds 
are used to derive the bound.
%Because $\mathcal{O}(n)$ such events exists, the error
%probability can be made \TODO{almost} arbitrarily large. 

\subsection{Summary of results}

We showed that the upper bound
on the required number of genes for ASTRAL* to reconstruct
the correct species tree with high probability grows as $f^{-2} \log n$. 
We further proved that there exist parts of the species
tree space where this upper bound is tight
in the sense that the probability of error can be bounded
strictly away from zero when the number of input gene trees 
is of order $f^{-2} \log n$. 
Thus, ASTRAL* requires
$\Omega(f^{-2} \log n)$ gene trees
to universally guarantee correct species tree reconstruction
with high probability.

\section{Simulations}
We now study the performance of ASTRAL-II and the NJst approach~\cite{njst} in simulations.
We seek to find 
the number of genes required by each method to recover the true unrooted species tree  with high probability.
We test ASTRAL-II, version 4.10.10~\cite{astral2}
and NJst as implemented in ASTRID~\cite{astrid}.

\subsection{Simulation procedures}
Testing sample complexity in simulations requires care. We seek the {\em minimum} number of genes with which methods stay under a specific level of error for specific model trees.

\subsubsection{Model species trees} 
We study three different tree forms, shown in Figure~\ref{Fig2}(a), all with exactly $n=8$ species and varying branch lengths (in coalescent units). 
The three forms we test are a fully unbalanced tree
where all internal branches  have length $f$, a fully balanced tree
where all internal unrooted branches have length $f$, and a balanced tree
where the two branches incident on the root have length 30 (i.e., very long) and the remaining branches have length $f$. We refer
to these trees as caterpillar, balanced, and double-quartet, respectively. 
We use nine values for $f$ in the range $[0.005,0.1]$, spaced so that $1/f^2$ is divided into equal chunks.

\begin{table}[!t]
\renewcommand{\arraystretch}{1.3}
\caption{Anomaly zone Conditions. For each each species tree topology and the smallest and the largest value of $f$, we
show the rank of the species tree topology
among the 10395 unrooted gene trees, its probability, and
the probability of the most likely gene tree other than the species tree. }
\label{tab:an}
\centering
% Some packages, such as MDW tools, offer better commands for making tables
% than the plain LaTeX2e tabular which is used here.
\begin{footnotesize}
\begin{tabular}{|l|c||c|r|r|}
\hline
Tree & $f$ & ST rank & $p(S)\times 10^{4}$&$p(g^t)\times 10^{4}$\\
\hline\hline
\textbf{Caterpillar}&\textbf{0.1} & \textbf{50}& $8.708$& $13.671$\\
%\textbf{Caterpillar}&\textbf{0.007}& \textbf{6004}& $0.742$&$2.551$\\\hline
\textbf{Caterpillar}&\textbf{0.005}& \textbf{6028}&$0.698$& $2.420$\\ \hline
Balanced&0.1 & 1& $28.096$& $23.170$\\%\hline
%Balanced&0.007& 1& $2.775$&$2.774$\\\hline
Balanced&0.005& 1&$2.5694$& $2.5690$\\ \hline
Double-quartet&0.1 & 1& $247.843$& $143.215$\\%\hline
%Double-quartet&0.007& 1& $130.564$&$125.178$\\\hline
Double-quartet&0.005& 1&$128.457$& $124.679$ \\ \hline
\end{tabular}
\end{footnotesize}
\end{table}

The unrooted anomaly zone requires
at least two consecutive short branches
(often defined as below 0.1)~\cite{Degnan2013}.
Based on calculations carried by the {\tt hybrid-coal}
software~\cite{hybrid-coal},% https://github.com/hybridLambda/ hybrid-coal
 the caterpillar tree is in the anomaly zone, but the balanced tree and the double quartet
trees are not (Table~\ref{tab:an}).
More importantly, most of the branch
lengths we have used are {\em extremely} short and
go beyond what one might expect in most real biological datasets. 
With $f=0.1$, assuming a diploid population size of 100,000, the branch would represent 20,000 generations (or roughly 100,000 years assuming a generation time of 5 years). 
In our simulations with $f=0.1$, the percent of quartets that agree with each branch (computed by ASTRAL-II~\cite{localpp})
is never more than 46\%. 
Thus, $f=0.1$ represents short branches, but is within ranges observed in biological datasets. On the other end, using the same assumptions, $f=0.005$ corresponds to a branch of only 1,000 generations, which is perhaps unrealistically short, and close to the range where biologists would consider a hard polytomy.
In our simulations,
all branches except the one long internal branch in the double-quartet tree (note that in the unrooted version, the double-quartet tree has only one long branch) result in 
high levels of discordance. 
For $f<0.1$,
the quartet support for species
tree branches is very close to 1/3 (i.e., random) and is never more than 36\%. 

Thus, we emphasize that our choice of branch lengths is motivated by a desire to explore the ability to
reconstruct the species tree under extreme conditions and establishing trends in data requirement;
we refer the reader to earlier publications for simulations seeking to emulate realistic conditions~\cite{astral,astral2}.

For each species tree, gene trees
are generated according to the MSC model
using the Dendropy package~\cite{dendropy}.

\subsubsection{Number of replicates and $\epsilon$} 
We seek to find the smallest number of gene
trees, $m$, so that the probability of incorrect
tree reconstruction by each method is no more than $\epsilon=0.1$.
%Thus, we seek to find $m$ such that ASTRAL has a 90\% probability of
%recovering the correct species tree. 
Computing the exact $\epsilon$ for a given $m$ in simulations requires an infinite number of replicates. In our simulations, for each species tree, we simulate 401 replicate datasets, with varying number of genes. 
To approximately achieve $\epsilon=0.1$, we search
for the smallest $m$ such that in no more
than 40 replicates each method outputs an incorrect species tree. Using normal approximation for binomial, 
observing 40 error events out of 401 replicates gives
us a 90\% confidence interval of $(0.075, 0.124)$ for $\epsilon$. 
Getting 90\% confidence intervals of $0.1\pm 0.01$ would require
close to 2,500 replicates, which is computationally infeasible given
the number of genes required per replicate.

\subsubsection{Binary search for $m$} 
Finding $m$
requires running ASTRAL-II and NJst with increasing numbers of genes
until reaching error levels below the specified threshold. Since our
choices of $f$ require tens of thousands of gene trees, this approach 
would become computationally impractical. Instead, we use a binary search
to find a tight window around the true value of $m$. The binary search starts with an initial range for $m$, picked as a guess by dividing results presented in Theorem~\ref{thm:upper} by 20. We set $m$ to the midpoint of the range and estimate $\epsilon$  by counting the number of replicates where each method outputs an incorrect tree. Based on the estimated $\epsilon$, we then decide to narrow down the range to the portion above or below the midpoint, and we repeat the process. We stop when the range becomes narrow enough (1/20 of our initial guess). We depict  final ranges and fit lines to their midpoints. Additionally, we check that upper and lower halves have been each chosen at least in one iteration of the binary search. This is to ensure that the true $m$ is within our initial range, and in occasions where this was not true, we expanded the initial range.

\subsection{Results} Results of our simulations are shown in Figure~\ref{Fig2}. 
As expected, decreasing $f$ results in
increased data requirements for both methods. 
The average number of genes required by ASTRAL-II ranges
from  206, 255, or 297 genes, respectively for the balanced, caterpillar, or double-quartet trees with $f=0.1$,
to $48948$, $59528$, and $93750$ genes with $f=0.005$. 
Matching our theoretical results, the number of  genes required by ASTRAL-II increased
proportionally to $1/f^{2}$. Similarly,
NJst seems to have data requirements
that scale with $1/f^{2}$ under these model conditions.
Interestingly, 
the tree shape had a substantial
and somewhat surprising
impact on the exact number of required genes.

\begin{figure}[tbp]
\begin{center}
\includegraphics[width=.43\textwidth]{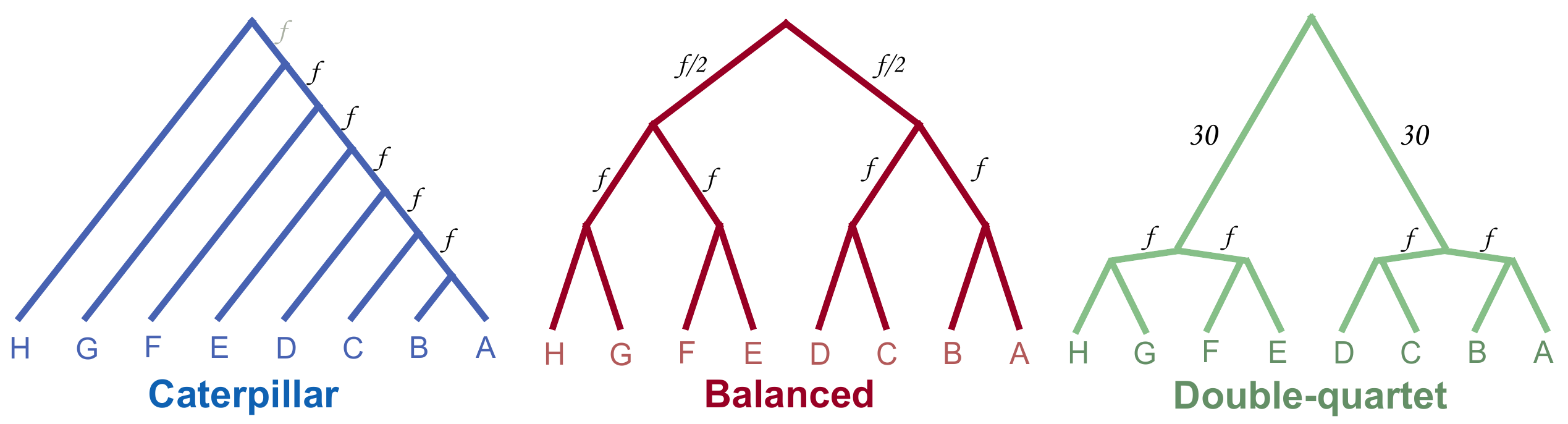}\\(a)\\
\includegraphics[width=.5\textwidth]{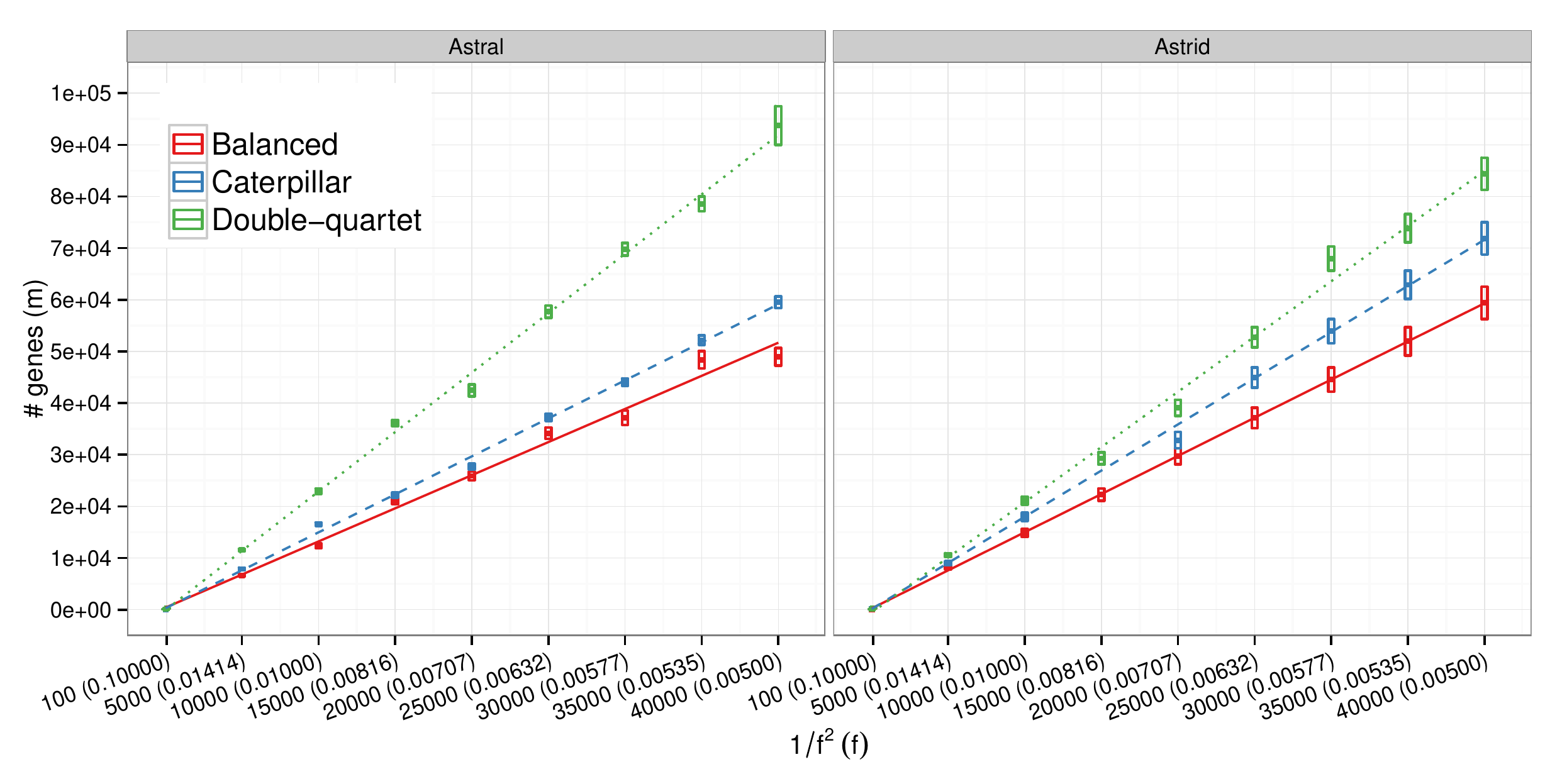}\\\vspace{-6pt}(b)\\
\includegraphics[width=.5\textwidth]{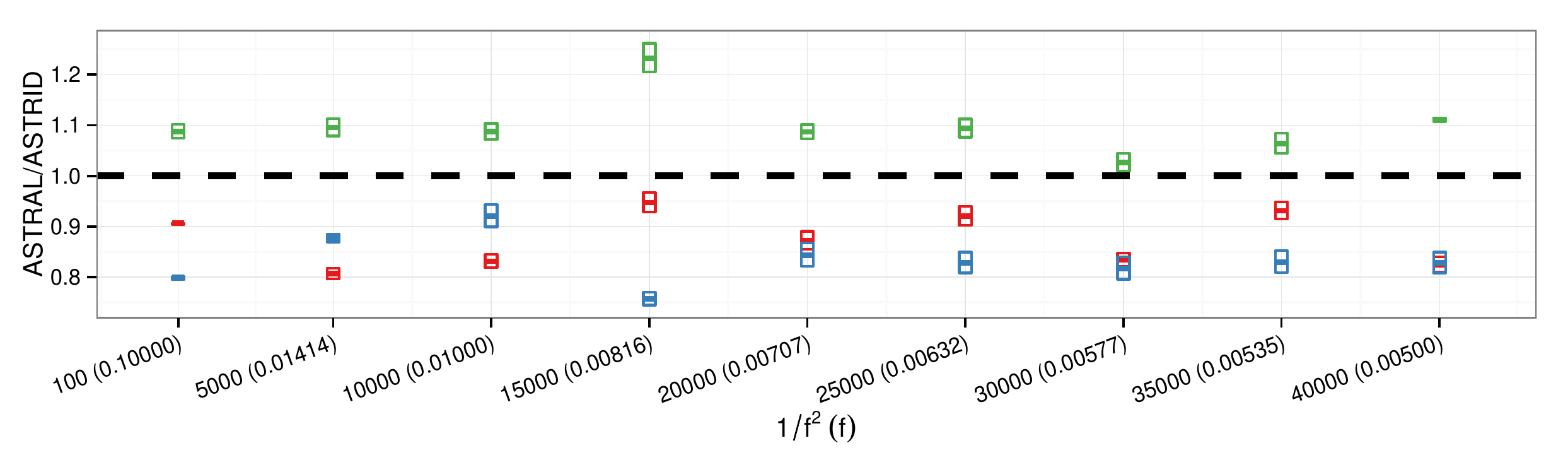}\\\vspace{-6pt}(c)
\end{center}
 \caption{Data requirement of ASTRAL-II and NJst (implemented in ASTRID) in simulations with $\epsilon=0.1$.
For each of the three different species tree shapes (panel a) and values of $f$, the length of the shortest branch, 
401 replicate datasets are simulated using
the MSC model, each with up to $10^5$ gene trees. 
%The x-axis shows $1/f^2$
%and indicates the corresponding $f$ parenthetically. 
(b) A binary search is used
to find an approximate range for the smallest number of genes with which ASTRAL-II (left) or NJst (right) recover the correct tree in at least 90\% of the 401 replicates. 
Boxes show these ranges and the horizontal lines
show the midpoint of each range.  A line
is fitted to midpoints for each tree shape.
(c) The ratio of the number of genes required by ASTRAL-II to the number required by NJst; below 1.0 indicates ASTRAL-II is better.
%The approximate binary search stops when the range
%of values gets to a number that is empirically determined
%but never exceeds 6.9\% of the estimated range.
% for reconstructing species trees with 8 species with an error rate of $\approx 0.1$.
 }
 \label{Fig2}
\end{figure}

With both methods, the caterpillar tree required slightly more genes
than the balanced tree, a result
that we don't find surprising as we will discuss. However, counter-intuitively, when
the first two branches in the balanced tree were made extremely large (i.e., the double-quartet tree),
the data requirement for both methods went up, and the increase was more dramatic for ASTRAL-II. For example, 
for $f=0.00816$, the balanced tree requires
 21,096 gene trees,
while the double-quartet requires 36,115,
an increase of about 70\%.
Comparing the data requirement of ASTRAL-II and NJst
reveals that ASTRAL-II typically required between 10 to 20\% fewer genes for the caterpillar and balanced trees, while NJst required fewer genes for the double quartet tree. 

The results show that the data requirement
of ASTRAL-II and its relative performance depend on the tree shape as well as on the exact pattern of branch lengths. We provide possible explanations for the surprising results in the Discussion section.
The results also suggest that the sample complexity of NJst might be
also quadratic with respect to $1/f$.

\section{Discussion}

%We provided theoretical bounds and simulation results that showed the number of genes required for reconstructing the species tree correctly with high probability using  ASTRAL* (i.e., exact ASTRAL) grows quadratically  with respect to the inverse of the smallest branch length in the true tree and logarithmically with respect to the numbers of species. 
The results we presented heavily rely on properties of the MSC model, which assumes that genes evolve in an {\em i.i.d.} fashion. This assumption, which is perhaps biologically unrealistic, enables us to treat true gene trees as draws from a multinomial distribution over the discrete space of all tree topologies, and
hence provide a nice mathematical framework for studying sample complexity. The main remaining difficulty in analyzing the MSC process is the fact the tree space is combinatorial and hard to enumerate. Luckily, using quartets provides a way to address this challenge because there are only three possibilities for an unrooted tree topology, and moreover, the anomaly zone does not exist. The analysis of ASTRAL* sample complexity for quartets mostly reduces to analyzing the concentration of a multinomial distribution with three probabilities $(1-\frac{2}{3}e^{-f},\frac{1}{3}e^{-f},\frac{1}{3}e^{-f})$; thus, we can use the well-established Hoeffding concentration inequality and the standard union bound technique to deal with dependencies for upper bounds. In the case of lower-bounds, even though the quartet case reduces to analyzing the CDF of binomials (Fig.~\ref{fig:binom}), the binomial CDF function has a complex dependence on the parameter $m$, which motivates us to use the Berry-Ess\'{e}en theorem (which essentially is a stronger version of the Central Limit Theorem) and also more strong bounds specific to binomials. Extending lower bound results from quartets to arbitrary $n$ is not quite as easy as the upper bounds because simple methods like union bounds are not available. Instead, we exploit the properties of the MSC model to design species trees with sufficiently long branches, and this makes the quartets independent conditioned on events of an arbitrarily high probability. To summarize, the key to obtaining our results is to reduce large $n$ cases to $n=4$, and to then turn the problem into the analysis of binomial distributions defined by MSC model.

%Several aspects of our analyses need further discussion. 

\subsection{Gene tree error}
Throughout the paper, we assumed a random sample of true gene trees generated by the MSC model is given. However, in practice, the 
gene trees are reconstructed from sequence data, and hence might contain errors. Moreover, processes
other than ILS may contribute to gene tree discordance. These violations of assumptions have the potential  to invalidate our results. One may suspect that when violations are not severe, many of our conclusions may remain intact. We can address this conjecture using a simplistic error model.

Assume that each gene tree quartet can be incorrectly estimated  with probability not exceeding some constant $e_0$, and also assume these events are independent across
different quartets. In the event of an error, each of the two alternative topologies is assumed
to be equally likely (i.e., has $1/2$ conditional probability). This model is not a realistic 
model of gene tree reconstruction error for a variety of reasons.
Quartets and their error events are not independent.
Moreover, the assumption of having unbiased
error is very strong, and does not necessarily follow from properties of 
tree reconstruction under any of the realistic
models of sequence evolution.  
Nevertheless, the following result provides insights about
robustness of our results. 

Let us denote by $Q_i$ the probability distribution according to the MSC model for the quartet $i$, 
i.e., $Q_i = \{p_i, r_i, r_i\}$ is the probability distribution on the set of the three possible topologies of the $i^{th}$ quartet. Then under the 
above error assumption, the observed quartets will have a probability distribution $Q_i' = \{p_i^{'},r_i^{'},r_i^{'}\}$,
with $p_i^{'} \geq p_i(1-e_0) + r_i e_0$, and $r_i^{'} \leq r_i(1-e_0/2) + p_ie_0/2$. As long as $p_i^{'} > r_i^{'}$, the upper 
bound for $m$ will have the same form as in Eq.\ref{eq:upper_bound}, but with larger constants. A sufficient condition for this
to happen is $p_i(1-e_0)+ r_ie_0> 1/3$ or $e_0 < 2/3$. Thus, with this particular error model, 
up to 66\% of gene trees can be incorrectly estimated and our sufficient conditions for correct reconstruction using ASTRAL* will not asymptotically change, indicating some level of robustness.

A rigorous analysis of the data requirement in the presence of gene tree error requires considering gene tree inference requirements. Since a sufficient condition for accurate gene tree inference is having genes with length that scales as $f^{-2}$~\cite{Erdos1999a}, one may expect that $m$ and the length of each gene both growing like $f^{-2}$ will be sufficient to bound error. This will result in total data requirement that scales with $f^{-4}$. 
The METAL~\cite{metal} method  that estimates species tree directly from sequence data has $f^{-2}$ overall data requirement
(using short constant-length genes), but it has remained theoretical and its performance on  data remains untested.

\subsection{ASTRAL-I and ASTRAL-II}
The constrained versions, ASTRAL-I and ASTRAL-II (see Table~\ref{tab:sum}), proceed similarly to ASTRAL*, with the added constraint that they restrict the species tree to draw its bipartitions from a fixed set ($X$). ASTRAL-I sets $X$ to the input set of gene trees, and ASTRAL-II further expands the set of allowable bipartitions using various heuristics. Clearly, any sufficient condition for ASTRAL-I to recover the true species tree will also be sufficient for ASTRAL-II, but necessary conditions for 
ASTRAL-II need not be as restrictive as ASTRAL-I. 

Analyzing the data requirement of ASTRAL-I requires answering the following question: how many genes are needed before every bipartition of the species tree has an arbitrarily high probability of appearing in at least one gene tree. To our knowledge, a recent result by Uricchio {\em et al.} is the only work on this ``bipartition cover'' problem~\cite{rosenberg2016}. Uricchio {\em et al.} give upper bounds on the required number of genes and, as our Appendix~\ref{sec:p-bpcover} shows, their upper bound correspond to an $m$ that increases as $\Theta(f^{-(n-3)})$. If that many genes were to be proved also {\em necessary}, we would conclude that ASTRAL-I and ASTRAL* have dramatically different data requirements. However, these analyses are very conservative and likely give only very loose upper bounds; the data requirement of the bipartition cover problem may prove to be much less. 
Tighter bounds for the bipartition cover problem 
are challenging to derive and need to be addressed in future work. 

\subsection{Discussion of simulation results}

%Since decreased branch lengths result in increased gene tree discordance, one might expect that required numbers of genes increases as branch lengths decrease. While our theoretical and simulation results supports this intuition for the {\em shortest} branch length in the species tree, our simulations reveal more complexity; 
In simulations, we found the species tree to be consequential for the required number of genes for both methods. 
Understanding the exact impact of tree shape would require
a more extensive study. Here we provide some conjectures. 
%Important insights about the strengths and weaknesses of ASTRAL are revealed in a closer look. 

The fact that caterpillar trees required more genes compared to balanced trees is not due to increased discordance. For example, both trees had a quartet score of 44\% with $f=0.1$
or 34\% with $f=0.005$ (computed based on 2000 simulated gene trees), indicating identical levels
of gene tree discordance. However, differences in data requirements 
of ASTRAL-II may
be explained by the dependency structure of quartets. 
We say a quartet is defined around an internal branch of a tree if the middle path of the quartet corresponds exactly to that branch. 
Take the ``centroid'' branch in each tree that divides the leaves into two sets of size four.
In the caterpillar tree, there are nine quartets around the centroid branch, and all these nine quartets share at least two leaves (E and D). In the balanced tree, there are sixteen quartets around the centroid branch, and eight pairs of quartets share no leaves in common. Thus, in the caterpillar case, we have fewer quartets and quartet trees are highly dependent, providing little {\em additional} information with respect to each other. 
On the other hand, in the balanced tree,
we have more quartets and many of them share few or no taxa and so, intuitively, may be less dependent. Note that having less dependence in the quartets around a branch can potentially help ASTRAL-II because even if the binomial associated with one of the quartets happens to give a score below $m/3$ for the species tree resolution, the other quartets can potentially make up for it. In other words, having fully independent quartets around a branch is like having more genes as input; balanced trees have more quartets and less  dependence among quartets compared to caterpillar trees.
%and this may explain why balanced trees are reconstructed correctly with fewer genes. 

We also observed that the double-quartet tree requires many more genes than the balanced tree. This is counter-intuitive because the longer branch in the double-quartet tree reduces gene tree discordance. 
In our simulations, 
the model balanced tree had quartet
scores of 44\% and 34\% for $f=0.1$
and $f=0.005$ while the double-quartet tree had 
scores of 71\% and 68\% (computed based on 2000 simulated gene trees). 
Why does the double-quartet tree require more genes despite having dramatically less  discordance?

The pattern may be related to quartet dependencies. 
The long branch at the root of the double-quartet tree results in a high likelihood that all lineages coalesce before reaching the root. For example, take the parent of $(H,G)$; quartets
around that branch include $H$, $G$, $E$ or $F$, and one of $\{A,B,C,D\}$ (Fig.~\ref{Fig2}); all quartets that include $F$ will all have the same exact frequency (same for $E$), hence providing no additional information with respect to each other. Thus, the long branch at the top is causing the eight quartets around each of the four lower branches to have very redundant information; this redundancy is unhelpful, as we discussed earlier. 
Thus, counter-intuitively, reduced discordance at the basal branch may also be resulting in more uncertainty about resolutions of the branches closer to the tips. The double-quartet tree that looks easier in terms of gene tree discordance turns out to be more difficult for ASTRAL-II to resolve correctly. Our results do not imply that the long branch itself is hard to recover. It simply indicates that it may complicate the recovery of short branches. This pattern, if observed more broadly, would incidentally resemble the long branch attraction problem in sequence-based analyses. 
We note that the conditions where ASTRAL-II had surprisingly high data
requirements resemble the types of trees that we invoked to prove the lower bounds for
arbitrary $n$. 
%In that sense, it could be said that the proofs provide
%further insight as to why this particular tree shape is hard for ASTRAL to 
%correctly reconstruct. 

Finally, our simulations start to shed light on differences between conditions where ASTRAL-II and NJst can each be expected to work better. In model conditions where all branches where short, ASTRAL-II outperformed NJst, especially with caterpillar trees. Both methods were sensitive to the presence of long branches, but ASTRAL-II more so. More broadly, our research indicates that a thorough examination of tree shapes for which ASTRAL and NJst perform well is required. Finally, the simulation results lead to the conjecture that NJst also requires $\mathcal{O}(f^{-2})$ gene trees; we hope future work will prove or disprove this conjecture.

\subsection{Practical implications}

Asymptotic results are not meant to provide
practitioners with means to predict how many gene trees
are needed. Two issues hamper such a goal. On the one hand, 
some of the parameters of the species tree (e.g., $f$) are not known in advance. On the other hand, the lower and upper bound values match only asymptotically and in practice, the effect of the unknown constants is important. 
Practitioners trying to decide the number of required
genes in advance may be only marginally helped by our results.
For example, given a large number of loci sampled from a few species (e.g., those with full genomes sequenced), they can empirically
estimate the number of required genes (e.g., judging accuracy
by examining local posterior probabilities generated by ASTRAL~\cite{localpp}). 
Then, to predict the number of required genes if they were
to sample more species, they could extrapolate using our results;
this would still need guessing the reduction in $f$ due to 
increased sampling and
such estimates will have to be taken as ballpark guesses.

Sample complexity results are useful to judge the effectiveness of methods.
The best possible sample complexity for the problem of inferring species trees from gene tree topologies is not known. If future
work proves $\mathcal{O}(f^{-2})$ to be the best possible,
our results will give practitioners reasons to have confidence in the efficiency of ASTRAL. On the other hand, if some other tool is
ever proved to have a better sample complexity, 
such asymptotic results should give practitioners {\em some} indication that alternative algorithms may be preferable. 
In either case, asymptotic results are not a replacement for
careful empirical and simulation studies and should be used as
a complementary source of information when judging methods.

\section*{Acknowledgements}

The authors thank Erfan Sayyari and anonymous reviewers
for helpful comments. The work was supported by 
National Science Foundation (NSF) grant IIS-1565862 to SM and SS
and NSF grant DMS-1149312 (CAREER) and NSF grant DMS-1614242 to SR.
Computations were performed on the San Diego Supercomputer Center (SDSC) and through the Extreme Science and Engineering Discovery Environment (XSEDE), supported by NSF grant number ACI-1053575.

\bibliographystyle{plain}
\bibliography{citations}

\appendix

\section{Proofs}
 
\subsection{Notation and definitions}
\label{subsec1}
We denote the number of species by $n$ and the number of gene trees by $m$. Let $k={n \choose 4}$ denote the total number of quadruples of leaves,
also known as quartets, in the species tree. We use $[k]$ to represent the set 
$\{1,2 \ldots ,k  \}$.

For each quartet $i \in [k]$, let $q_i$ represent the unrooted tree topology induced by the species tree for quartet $i$. There are three possible unrooted tree topologies that can be induced by  gene trees and we represent these three configurations by $q_{i1}, q_{i2}$, and $q_{i3}$. We choose to use $q_{i1}$ to denote the configuration of the species tree, i.e., $q_{i1} = q_i$. Also let $n_{ij}$, for $j=1,2,3$, represent the number of gene trees that induce the $q_{ij}$ topology. By definition, $\sum_{j=1}^3 n_{ij} = m$ for all $i \in [k]$. 

Under the MSC model, for each $q_i$, the probability that a gene tree has 
configuration $q_{i1}$ is $p_i = 1 -\frac{2}{3}e^{-d_i}$ where
$d_i$ is the length of the middle path in the quartet topology of the species tree measured in  coalescent units~\cite{Allman}.
The other two configurations have equal probability $r_i = (1-p_i)/2$.
Let us denote by $\delta_i$ the quantity $p_i - 1/3$ and let $\delta$ be the minimum over all $\delta_i$s. All quartets defined around the shortest branch in the species tree have $\delta_i=\delta$ and we use $f$ to denote the length of this shortest species tree branch. 

ASTRAL* returns the species tree that
maximizes the score $s(T) = \sum_{i=1}^{k} n_{i,t_i(T)}$ where $t_i(T)$ gives the index of the topology for quartet $i$ found in tree $T$.
Hence observe that
the score of the true species tree
is $\sum_{i=1}^{k} n_{i1}$.

\subsection{Proof of Theorem \ref{thm:upper}}
\label{sec:p-upper}
\noindent{\bf Theorem 2.1.}
{\em
	Consider a model species tree with minimum branch length $f< \log(\sqrt{2})$. Then, for any $\epsilon>0$, ASTRAL* returns the true species tree
	with probability at least $1-\epsilon$ if the number of input error-free gene trees satisfies 
	\begin{equation}
	\label{eq:upper_bound}
	m > \frac{9}{2}\log\bigg(\frac{4{n \choose 4}}{\epsilon}\bigg)\frac{1}{(1-e^{-f})^2}.
	\end{equation}
}

To prove this result, we require an upper bound on the probability of error in terms of $m$, $n$ and $f$. For that, we observe that a sufficient condition for ASTRAL* to return the correct species tree is that for all the quartets $i \in [k]$, the true topology is observed in a majority of the $m$ gene trees, i.e., $n_{i1}>\max\{n_{i2},n_{i3}\}$. Thus, for ASTRAL* to make an error, at least one quartet must have an alternate topology in a majority of the gene trees. We upper bound the probability of this event for one quartet by using Hoeffding's inequality, and then take a union bound over all quartets to bound the overall probability of error.

Let $w < \delta/2$.
We use $A_i$  to refer to
the event $n_{i1}/m \in (p_i-w,p_i+w)$, and use $B_i$ and $C_i$ to refer to events $n_{i2}/m\in(r_i-w,r_i+w)$ and $n_{i3}/m\in(r_i-w,r_i+w)$, respectively. 
\begin{lemma}\label{lemma:event}
For any $w < \delta/2$,
ASTRAL* is guaranteed to return the true
species tree if 
the event $D_i = A_i\cap(B_i \cup C_i)$ occurs for all $i \in [k]$.
\end{lemma}

\begin{proof}
Recall the definition of the score $s(T)$ and observe that
ASTRAL* returns the true species tree when the following condition holds:
$$n_{i1} > \max\{n_{i2},n_{i3}\}, \qquad \forall i\in [k].$$
We claim that, for all values of $w$ smaller than
$\delta/2$, we can ensure that $n_{i1}$ is the largest among $n_{ij}$, when the event $D_i$ occurs. 
Indeed, let $D_i$ hold and assume w.l.o.g.~that $B_i$ occurs. In the worst case, $n_{i1} = (p_i - w)m$ and $n_{i2} = (r_i-w)m$, which
results in $n_{i3} = (r_i+2w)m$ because
the sum of all three $n_{ij}$s is $m$.
In this case, we need
$r_i+2w=(1-p_i)/2 + 2w < p_i -w$, or after multiplying by $2$ and rearranging
$3p_i - 1 > 6w$.
This is true for all $w < \frac{1}{2}\delta_i = \frac{1}{2}(1/3 - p_i)$. 
\end{proof}

\begin{lemma}\label{lemma:upperprob}
Suppose the input to ASTRAL* is $m$ true gene trees generated under the MSC model. For every $\epsilon >0$, ASTRAL* reconstructs the correct species tree with probability at least $1-\epsilon$ if
 \begin{equation}\label{eq:lem2}
  m > 2\log\bigg( \frac{4{n \choose 4}}{\epsilon}\bigg)\frac{1}{\delta^2}.
 \end{equation}
%  \textcolor{blue}{ I have restated the result with a $>$ instead of $\geq$, by taking the limit of $\alpha \to 1$. This is because we want this bound to be as small as possible, and thus taking $\alpha \to 1$ gives us the best constants. Setting $\alpha$ to $3/4$ would have given us a leading constant of 32/9 instead of the 2, with a $\geq$. }
\end{lemma}

\begin{proof}
 If $E$ denotes the event that ASTRAL* returns the wrong species tree, then by Lemma~\ref{lemma:event} we have by a union bound
 \begin{align}
 \label{eq:err1}
  P(E) 
  &\leq P\big( \cup_{i=1}^{k} D_i^c\big)\\
   & = P( \cup_{i=1}^k A_i^c\cup(B_i^c\cap C_i^c)    )\\
   &\leq P\big( \cup_{i=1}^{k}(A_i^c\cup B_i^c)\big)\\ 
  &\leq \sum_{i=1}^k \left[P(A_i^c)+P(B_i^c)\right]
 \end{align}
Thus we need a bound on $P(A_i^c)$ and $P(B_i^c)$. Let us define a random variable $X_{ig}$, for $i \in [k]$ and $g \in [m]$, 
 which takes the value 1 if the configuration of the $i^{th}$ quartet in the $g^{th}$ gene tree is the same as $q_i$;
 and is zero otherwise. Then $(X_{ig})_{g=1}^m$ is a sequence of independent Bernoulli random variables with parameter $p_i$
 and $n_{i1}=\sum_{g=1}^{m} X_{ig}$. 
 The event $A_i^c = \{ n_{i1} \notin (m(p_i-w),m(p_i+w))\}$ with $w < \frac{\delta}{2}$ becomes increasingly unlikely as $m$ grows and its probability
  can be bounded using Hoeffding's inequality (see e.g.\cite{lugosi2004concentration}) to get
\begin{equation*}
    P(A_i^c) \leq 2e^{-2m w^2}.
\end{equation*}
The same bound holds for $P(B_i^c)$.
Plugging this into~\eqref{eq:err1}, we get
\begin{equation}\label{eq:mm}
P(E) \leq {n \choose 4}4e^{-2 m w^2}.
\end{equation}
To make the Right Hand Side (RHS) smaller than $\epsilon $, we take
\begin{align}  \label{eq:lem1}
m \geq \frac{1}{2 w^2}\log\bigg( \frac{4{n \choose 4}}{\epsilon}\bigg).
 \end{align}
The result follows from the fact that we can choose any $w < \delta/2$.\end{proof}

%This result follows directly by combining Lemmas \ref{lemma:upperprob} and \ref{lemma:alpha}. 

\begin{proof}[Proof of Theorem \ref{thm:upper}]
Lemma~\ref{lemma:upperprob} gives us the gene tree requirement in terms of the parameter $\delta$.
The result in Theorem~\ref{thm:upper} immediately follows by replacing $\delta$ with $\frac{2(1-e^{-f})}{3}$. In the limit of small $f$, we note further that $\frac{(1-e^{-f})^2}{f^2}$ increases to 1. This gives the $\mathcal{O}(f^{-2})$ dependence of the bound on $m$ as mentioned in the Remark after the statement of the theorem.\end{proof}

\subsection{Proof of Theorem
\ref{one_q}}
\label{sec:p-one_q}

\noindent{\bf Theorem 2.2.}
{\em For a species  tree with $n=4$ leaves, let $f$  be the length of the single internal branch in the unrooted species tree.
 Then, for any $\epsilon \in [0,0.5)$, there exists an $f_0 >0$ and  a constant $d_{\epsilon}$,  such that for all $m \leq d_{\epsilon}/f^2$ with $f \leq f_0$, the probability that ASTRAL* outputs the wrong quartet is at least $\epsilon$}. 
 %\end{theorem}

\begin{proof}

We divide the proof into two cases. First we show that the result holds for $m \leq c_{\epsilon}/f$, for some constant $c_{\epsilon}$. That follows from the fact that, in that case, a non-trivial fraction of genes coalesce only in the ancestral root population. We then show separately that the result also holds for $ c_{\epsilon}/f < m < d_{\epsilon}/f^2$, for some choice of $d_{\epsilon}$. There, we use the Berry-Ess\'{e}en theorem to control the probability that a majority of the genes return the wrong topology.

\paragraph*{Case 1 ($m \leq c_{\epsilon}/f$)} Let
$G$ be the event that, in all gene trees, all the lineages reach the root population without any prior coalescence.
Depending on the rooting, the quartet species tree can either have the symmetric topology or the asymmetric topology. In both cases we have 
\begin{align}\label{eq:prob-of-g}
P(G) = (e^{-f})^m = e^{-mf}.
\end{align}
(Recall that $f$ represents the internal branch length in the unrooted tree. Thus it equals the length of the single internal branch in the asymmetric case and it equals the sum of the two internal branch lengths in the symmetric case.)

Under $G$, for every gene tree, all three topologies are equally likely. So $n_{11}, n_{12}$ and $n_{13}$ have the same distribution. (There is only one quartet in this case.) Since ASTRAL* selects the output tree randomly in case of a tie, by symmetry we can conclude that $P(E|G) = 2/3$ and hence, by~\eqref{eq:prob-of-g}, 
\begin{equation*}
    P(E) \geq P(E|G) P(G) \geq \frac{2}{3}e^{-c_{\epsilon}}.
\end{equation*}
For all $\epsilon \in [0,0.5)$, the choice $c_{\epsilon} = \log(\frac{2}{3\epsilon})$ ensures finally that $P(E) \geq \epsilon$. 

\paragraph*{Case 2 ($ c_{\epsilon}/f < m < d_{\epsilon}/f^2$)} 
Recall that, when $n=4$, we let $F$ be the event that $\{n_{11} \leq m/3\}$. We write
\begin{equation*}
P(F)
=P(n_{11} \leq  m/3) 
=P( Y \leq y)
\end{equation*}
where $Y = \frac{n_{11}-pm}{\sqrt{p(1-p)m}}$ and $y = \frac{-m(p-1/3)}{\sqrt{p(1-p)m}}$. Then, by the Berry-Ess\'{e}en theorem (as stated e.g.~in~\cite{Durrett}), we have
\begin{align*}
 \bigg|P(Y\leq y) - \Phi(y)\bigg| &\leq \frac{cp(1-p)[p^2 + (1-p)^2]}{[p(1-p)]^{3/2}m^{1/2}}\\
&\leq \frac{cp}{[p(1-p)]^{3/2}m^{1/2}},
\end{align*}
where $\Phi(\cdot)$ denotes the normal CDF and $c$ is a universal constant (which can be taken to be $3$). 
Thus
\begin{align}
 P(F) =P( Y \leq y) \geq \Phi(y) - \frac{cp}{[p(1-p)]^{3/2}m^{1/2}}.\label{eq:berry-esseen}
\end{align}

It remains to show that the RHS above is bounded away from $0$ when $f$ is small. We estimate the two terms separately.
We know that as $f\to 0$, $p \to 1/3$ and $p(1-p) \to 2/9$, and $p(1-p)$ is monotonic for $0<p<1/2$. Thus, there exists an $f_1$ such that $p \leq 10/27$ and $p(1-p) \geq 1/9$ whenever $ f \leq f_1$. So, for all $f \leq f_1$,
\begin{align*}
 y = \frac{m^{1/2}(p - 1/3)}{(p(1-p))^{1/2}} & \leq  3m^{1/2}(p - 1/3),
\end{align*}
and
\begin{align*}
    \frac{cp}{ (p(1-p))^{3/2}m^{1/2}} & \leq \frac{10c }{m^{1/2}}.
\end{align*}
Thus we get the following relation 
\begin{equation}
\label{eq:eq_1}
P(F) \geq \Phi \big(-3m^{1/2}(p - 1/3)\big) - \frac{10c }{m^{1/2}}.
\end{equation}
Let us first consider the second term on the RHS. For $m > c_{\epsilon}/f$,
\begin{equation*}
    \frac{10c}{m^{1/2}} < \frac{10cf^{1/2}}{c_{\epsilon}^{1/2}}
\end{equation*}
Let $\bar{f}_2$ be such that $\frac{10c\sqrt{\bar{f}_2}}{\sqrt{c_{\epsilon}}} = \frac{\gamma}{2}$, for $\gamma = 1/2 - \epsilon$. For all $f \leq f_2= \text{min}\{\bar{f}_2, f_1\}$ and $m > c_{\epsilon}/f$, we have $10c/\sqrt{m} < \gamma/2$. 

As for the first term on the RHS of~\eqref{eq:eq_1}, it is decreasing in $m$ for a fixed $f$. Thus for any $f \leq f_2$, if we find a value of $m$ as a function of $f$, say $m_0(f)$, which is larger than $c_{\epsilon}/f$ and for which this first term is no smaller than $1/2 - \gamma/2$, then the required result also holds for all $m \in (c_{\epsilon}/f,m_0(f)]$. 

Let us set $m_0(f) = d_{\epsilon}/f^2$, where $d_{\epsilon}$ will be determined below. Then we have
\begin{align}
\label{berry_esseen1}
    P(F) &\geq \Phi \left(-3\sqrt{d_{\epsilon}}\left(\frac{2}{3}\left( \frac{1-e^{-f}}{f}\right)\right)\right) - \frac{\gamma}{2} \\
        & = \Phi \left(-2\sqrt{d_{\epsilon}}\right) - \frac{\gamma}{2},
\end{align}
where we used that for $f >0$, $\frac{1-e^{-f}}{f} \leq 1$. Now, if we set 
\begin{equation}
d_{\epsilon} = \bigg(\frac{\Phi^{-1}\left(1/2 - \gamma/2\right)}{2}\bigg)^2
\end{equation}
we get that, for all $f \leq f_2$ and $m \in (c_{\epsilon}/f, d_{\epsilon}/f^2]$, $P(E) \geq 1/2 - \gamma = \epsilon$. Finally, we define $f_0 = \text{min}\{f_2, d_{\epsilon}/c_{\epsilon} \}$ to ensure that the above interval is non-empty. 
\end{proof}

\subsection{Proof of Claim~\ref{simple_1}}
\label{sec:p-simple_1}
\noindent{\bf Claim~\ref{simple_1}.}
{\em
For any $n$ and $\epsilon \in [0,0.5)$, there exists a species tree with $n$ leaves and shortest
branch length $f$ such that when ASTRAL* is used with  $m \leq d'_{\epsilon}/f^2$ gene trees, for some constant $d'_{\epsilon}$, the probability that ASTRAL* reconstructs the wrong tree
is at least $\epsilon$. 
}

\begin{proof}
%We need to show that for any given $\epsilon \in [0,0.5)$, there exists a species tree for which the reconstruction error with ASTRAL* will be at least $\epsilon$ if the number of input gene trees is less than $D_{\epsilon}/f^2$ for some constant $D_{\epsilon}$.

The idea is to reduce the proof to the case of Theorem~\ref{one_q} with an appropriate choice of species tree.
We consider a species tree with the topology $(((a,b),c),\mathcal{X})$, where $\mathcal{X}$ is an arbitrary tree with $(n-3)$ leaves, which we denote by $X = [n-3]$ for convenience. We denote the edge above $(a,b)$
by $e$ and set its length to $f$. We also set the length of the the branch above $\mathcal{X}$ to $g$, a large value which we will specify later.
%We proceed in the following steps:
%\begin{itemize}
%\item 

For $x\in X$, let
$n_{x1}, n_{x2}, n_{x3}$ be the frequencies 
of the quartet trees congruent with
the unrooted topology of $(((a,b),c),x)$,
$(((a,c),b),x)$, and $(((a,x),b),c)$ respectively.
When $n_{x1} < m/3$ for all $x \in X$,
ASTRAL* cannot output the true tree.
To see this, note that, under that condition, $\sum_{x\in X} n_{x1}<m(n-3)/3$; thus, $\max\{\sum_{x \in X} n_{x2},\sum_{x \in X} n_{x3}\}>m(n-3)/3$. W.l.o.g.~assume that $\sum_{x \in X} n_{x2}>m(n-3)/3$.
Then a tree
that includes the branch $e$ cannot have the maximal score because the unrooted topology of $(((a,c),b),\mathcal{X})$ will have a higher score. Indeed, observe that all quartets not of the form $\{a,b,c,x\}$ for some $x \in X$ contribute equally to the unrooted topologies of $(((a,b),c),\mathcal{X})$ and $(((a,c),b),\mathcal{X})$.

Let $H_i$ denote the event that, in gene $i$, all lineages from $X$ coalesce before reaching the root branch (i.e., within or below the long branch with length $g$) and let $H = \cap_{i=1}^m H_i$. Under the event $H$,  $n_{xj}=n_{yj}$ for all $x, y\in X$ and all $j \in \{1,2,3\}$; hence $P(\cap_{x \in X} n_{xj} <  m/3|H)
=P(n_{1j} <  m/3|H)$.
Thus, $\{n_{1j} <  m/3\}$ is sufficient to produce the wrong topology. We now bound the probability of that event.
%Thus, under $H^m$ we can use $n=$ results
%from 
% Theorem~\ref{one_q}. 
%the probability that all quartets in $Q^e$ have a certain frequency is equal to the probability that one of them has that frequency. 

Let us define 
$$
\epsilon_1 
= \frac{1}{4} + \frac{1}{2} \epsilon,
= \epsilon + \frac{1}{2}\left(\frac{1}{2}-\epsilon\right) > \epsilon
$$
and 
$$
\eta = 1-\epsilon/\epsilon_1.
$$
From Theorem~\ref{one_q}, we know that there exists an $f_0$ such that, for all $f \leq f_0$, if $m \leq d_{\epsilon_1}/f^2$, then $P(n_{1j} <  m/3|H)\geq \epsilon_1$. Let 
$$
m_0 = m_0(f) =  d_{\epsilon_1}/f^2,
$$
as defined in the proof of Theorem~\ref{one_q}. If $E$ denotes the event that ASTRAL* produces the wrong output, then we have 
\begin{equation}
\label{eq:claim_eq1}
P(E) \geq P(n_{1j} <  m/3|H) P(H)\geq P(H)\epsilon_1.
\end{equation}
Proving a lower bound of $1-\eta$ for $P(H)$ then implies the result, by our choice of $\eta$. For this, we choose appropriately the value of $g$. Specifically, we pick $g$ large enough to ensure that $P(H_i) \geq 1-\eta_1$ for $\eta_1 = 1-(1-\eta)^{1/m_0}$. Due to the independence of the genes, we get that for all $m \leq m_0$
%{\bf SM: shouldn't the first equality be an inequality?}
\begin{equation}
P(H) \geq (1-\eta_1)^m \geq (1-\eta_1)^{m_0} = 1-\eta. 
\end{equation}
Combining this with~\eqref{eq:claim_eq1}, we get that $P(E) \geq \epsilon$. Thus, we have shown the existence of a tree with $n$ leaves for which ASTRAL* is wrong with probability at least $\epsilon \in [0,0.5)$ if the number of input gene trees is less than or equal to $d'_{\epsilon}/f^2$, where $d'_{\epsilon} = d_{\epsilon + 1/2(1/2-\epsilon)}$.\end{proof}

\subsection{Proof of Theorem~\ref{thm:upperfull}}\label{sec:p-upperfull}
%Starting from results of Claim~\ref{simple_1}, We proceed in two steps:
\noindent {\bf Theorem~\ref{thm:upperfull}.} {\em
For any $\rho \in (0,1)$ and $a \in (0,1)$, there exist constants $f_0$ and $n_0$ such that the
following holds. For all $n \geq n_0$ and $f \leq f_0$,
there exists a species tree with $n$ leaves  and shortest
branch length $f$ such that when ASTRAL* is used with $m \leq \frac{a}{5}\frac{\log n }{f^2}$ gene trees, the event $E$ that ASTRAL* reconstructs the wrong tree has probability
\begin{equation}
\label{eq:new_lower}
P(E) \geq 1-\rho. 
\end{equation}
}

\begin{proof}
To simplify the presentation, we assume that $n$ is divisible by three. (It is straightforward to generalize the argument.)
Consider a species tree constructed as follows.  
Make $n/3$ rooted triplet trees all of them with an internal
branch of length $f$ and arbitrarily index the triplet trees by $i \in [n/3]$. Then, pick any arbitrary tree with  $n/3$ leaves labeled by $[n/3]$
and all branch lengths equal to exactly $g$. 
Finally, connect each leaf of this arbitrary tree with the root of
the corresponding triplet tree. 
Let $J_{i}$ denote the event that, in all gene trees, all three lineages of the triplet $i$ coalesce in its root branch. Let $J = \cap_{i=1}^{n/3}J_{i}$.
We proceed in two steps:
\begin{enumerate}
\item First we introduce a result analogous to Theorem~\ref{one_q} with
an extra $\log n$ factor in the number of genes. This gives us a smaller probability of error of the form $n^{-a}$ for some $0<a<1$ on a single quartet.

\item Then we consider the species tree above in which there are $\Omega(n)$ chances for a quartet-based error in ASTRAL* to occur and show that at least one of these happens with sufficiently large probability. 
\end{enumerate}

\paragraph*{Step 1} Let us consider one of the  triplets above $i\in[n/3]$, with a topology $((u,v),w)$, as well as a leaf $x$ outside of the triplet. Let $n_{i,1}$ be the number of genes with the correct species tree topology and let $F_i$ denote the event that $\{n_{i,1} \leq m/3\}$. For reasons that will be made clear in Step 2, we are interested in $P(F_i|J_i)$. Conditioned on $J_i$, let $p'_g$ and $r'_g = (1/2)(1-p'_g)$ be the probabilities of observing $((u,v),(w,x))$ and $((u,w),(v,x))$ or $((v,w),(u,x))$ respectively in the unrooted gene trees. As $g \to +\infty$, $p'_g \downarrow 1 - (2/3)e^{-f}$. Let $g_0$ be large enough (as a function of $f$) so that, for all $g \geq g_0$, $p'_g - 1/3 \leq (2.1/3)(1-e^{-f})$. Assume from now on that $g \geq g_0$.

To deal with the extra $\log n $ factor in the number of genes, we revisit~\eqref{eq:berry-esseen}. This time, instead of Berry-Ess\'{e}en, we use a tailored estimate on the binomial.
We will need the following bounds.
The Kullback-Liebler divergence between two Bernoulli random variables satisfies
\begin{align*}
D(a\|b) 
&= a\log \frac{a}{b} + (1-a)\log \frac{1-a}{1-b}\\
&\leq a\left(\frac{a}{b}-1\right) 
+(1-a)\left(\frac{1-a}{1-b}-1\right) \\
&= \frac{(a-b)^2}{b(1-b)},
\end{align*}
where we used
the fact that $\log x \leq x-1$.
By~\cite[Lemma 4.7.1]{Ash},
we also have 
$$
{m \choose k}p^k(1-p)^{m-k}
\geq \frac{\exp\left(-m D\left(k/m\middle\|p\right)\right)}{\sqrt{8 m (k/m) (1-k/m)}}.
$$

Using these two bounds we get
\begin{align*}
P(F_i|J_i) 
&= \sum_{k=0}^{m/3}{m \choose k}(p'_g)^k(1-p'_g)^{m-k} \\
&\geq \sum_{k=m/3 - \sqrt{2m}}^{m/3}{m \choose k}(p'_g)^k(1-p'_g)^{m-k} \\
&\geq \sum_{k=m/3 - \sqrt{2m}}^{m/3}
\frac{\exp\left(-m D\left(k/m\middle\|p'_g\right)\right)}{\sqrt{8 m (k/m) (1-k/m)}} \\
&\geq \sum_{k=m/3 - \sqrt{2m}}^{m/3}
\frac{1}{\sqrt{2m}}\exp\left(-m D\left(k/m\middle\|p'_g\right)\right) \\
&\geq \sum_{k=m/3 - \sqrt{2m}}^{m/3}
\frac{1}{\sqrt{2m}}\exp\left(-m\frac{(p'_g - k/m)^2}{p'_g(1-p'_g)} \right)\\
&\geq \exp\left(-m\frac{(p'_g - 1/3 + \sqrt{2/m})^2}{p'_g(1-p'_g)} \right),
\end{align*}
where we used that $p'_g > 1/3$ on the last line.
Above we assumed, to simplify the notation, that $\sqrt{2m}$ is an integer. Using $p'_g - 1/3 \leq (2.1/3)(1-e^{-f})$, $1-e^{-f} \leq f$ 
and $m \leq m_1(f) = (a/5) f^{-2} \log n$, we get
\begin{align*}
&P(F_i|J_i)\\
&\ \geq \exp\left(-\frac{(a/5) \log n}{f^2} \cdot \frac{\{2.1 f/3 + f/\sqrt{(a/10) \log n}\}^2}{p_g'(1-p_g')} \right).
\end{align*}
Choose constants $f_0$ and $n_0$ depending on $a$ ensuring that $f \leq f_0$ and $n \geq n_0$ imply that $p_g'(1-p_g') \geq 1/9$ and that the expression in curly brackets in the previous display is less than $f\sqrt{5/9}$.
Then
\begin{align}
\label{eq:nbet}
P(F_i|J_i)
&\geq
n^{-a}.
\end{align}

\paragraph*{Step 2}
% Let us consider a triplet $\tau$ and  represent the rest of the tree by $\mathcal{X}_{\tau}$.
% Under the event $K_{\tau}=\cap_{i\ne\tau}J_i$, we can use 
% the bound in~\eqref{eq:nbet}.
For each $i \in [n/3]$, let $E_{i}$ denote the event that all quartets containing the triplet $i$ have an alternative topology (i.e., different from the species tree)
in the majority. We note that, when $J_{i}$ occurs, then $F_{i}$ implies $E_{i}$.
If $E$ denotes the event that ASTRAL* outputs a wrong tree, then a sufficient condition for an error to occur is that at least one of the $E_{i}$ events occurs. 
Under the MSC, the events $\{F_{i} \cap J_{i} : i \in [n/3]\}$ are independent. 
Thus,
\begin{align}
P(E) 
&\geq 1 - \prod_{i \in [n/3]}[1- P(F_{i} \cap J_{i})]\nonumber\\ 
&\geq  1 - \prod_{i \in [n/3]}[1- n^{-a} P(J_{i})],\label{eq:e-aux1}
\end{align}
by~\eqref{eq:nbet}.
It remains to bound $P(J_{i})$.

% Thus,
% \begin{equation*}
% P^m(E|J) \geq \bigg(1- \bigg(1-\frac{n^{-a}}{4\beta(n)}\bigg)^{n/3}\bigg) 
% \end{equation*}
% We have:
% \begin{align}
% P(E_\tau|K_\tau) &= P(J_\tau)P(E_\tau|J)+P(J_{\tau}^c)P(E_\tau|K_{\tau},J_{\tau}^c)\\
%               & \leq P(E_\tau|J)+ \rho_1 \label{eq:neq2}
% \end{align}
% where $\rho_1$ is an upper bound on $P(J_{\tau}^c)$ to be specified below. 
Let $1-\rho_2$ be the probability that the full coalescence for triplet $i$ occurs in its root branch in one gene tree. Then $P(J_{i}) = (1-\rho_2)^{m} \geq (1-\rho_2)^{m_1(f)} $ for $m \leq m_1(f)$. Thus, for any $\rho_1 \in (0,1)$, we can choose $g_1$ large enough to ensure that $g \geq g_1$ implies $\rho_2 \leq 1-(1-\rho_1)^{1/m_1(f)}$ and therefore $P(J_{i})\geq 1- \rho_1$. Assume from now on that $g \geq \max\{g_0, g_1\}$.

Going back to~\eqref{eq:e-aux1}, using the fact that $(1-\frac{1}{x})^x \leq e^{-1}$, we have 
\begin{align*}
P(E) 
&\geq  1 - [1- n^{-a} (1-\rho_1)]^{n/3}\\
&\geq  1 - \exp\left(
- \frac{(1-\rho_1) n^{1-a}}{3}
\right)\\
&\geq 1-\rho,
\end{align*}
for $n \geq n_1$, for an appropriate $n_1$ depending on $a$.

\end{proof}

\section{Bipartition cover bounds}\label{sec:p-bpcover}

%\begin{definition}
A set of gene trees is called a \emph{bipartition cover} of a species tree $S$ if the set of bipartitions included in the gene trees includes all the bipartitions of $S$. 
%\end{definition}
In \cite{rosenberg2016}, the authors show that input gene trees form a bipartition cover with probability at least $q$ if 
%$m \geq M_{n,q}$, where $M_{n,q}$ is defined as:
\begin{equation*}
m \geq M_{n,q} = \frac{\log((1-q)(n-3))}{\log(1-g_{n-2,1}(f))} =\frac{ C_{n,q}}{-\log(1-g_{n-2,1}(f))},
\end{equation*}
where $g_{n-2,1}(f)$ is the probability that $n-2$ lineages coalesce to 1 within time $f$ and $C_{n,q}$ is implicitly defined above. 
When $n$ is fixed, we claim that $M_{n,q}$
is $\Theta(f^{-(n-3)})$ as $f \to 0$.

%Let us try to get an idea of how large this term $M_{n,q}$ is in terms of $f$. 
Note that $g_{n-2,1}(\cdot)$ is the CDF of a random variable $T_{n-2}=\sum_1^{n-3}W_i$ where the $W_i$s are independent exponential waiting times with parameters $\lambda_i = {n-1-i\choose 2}$. Thus,
\begin{equation}
\label{eq:heuristic1}
P(T_{n-2} \leq f) \leq \prod_{i=1}^{n-3}P(W_i \leq f)
%\\ & =\prod_{i=1}^{n-3}(1-e^{-\lambda_if}) 
\leq (\prod_{i=1}^{n-3}\lambda_i)f^{n-3}.
\end{equation}

Since we have $ \lim_{f \downarrow 0} \frac{1-e^{-\lambda_i f}}{\lambda_i f}=1$,   for $f$ small enough we have $ \frac{1-e^{-\lambda_i f}}{\lambda_i f} \geq 0.9^{1/(n-3)}$ for all $i$. So we can write
\begin{align*}
P(T_{n-2} \leq f) 
&\geq \prod_{i=1}^{n-3}P\left(W_i \leq \frac{f}{n-3}\right)\\ &\geq  0.9\bigg(\prod_{i=1}^{n-3}\lambda_i\bigg)\bigg(\frac{f}{n-3}\bigg)^{n-3}.
\end{align*}
Thus for small $f$ and fixed $n$, the probability $g_{n-2,1}(f)$ scales as $f^{n-3}$.
%We can use this to get an idea of how the term $M_{n,q}$ depends on $f$. First,

We now note that
\begin{equation}
\label{eq:heuristic2}
M_{n,q} = C_{n,q}\frac{1}{-\log(1-g_{n-2,1}(f))} \stackrel{(a)}{\leq} \frac{C_{n,q}}{g_{n-2,1}(f)},
\end{equation}
where $(a)$ follows from $\log(x) \geq 1-1/x$ for positive $x$.
Also, by using $\log(x) \leq x -1$ and assuming that $f$ is small enough to ensure that the upper bound on $g_{n-2,1}$ in~\eqref{eq:heuristic1} is smaller than 0.9, we have 
\begin{equation}
\label{eq:heuristic3}
M_{n,q} \geq C_{n,q}\frac{1-g_{n-2,1}(f)}{g_{n-2,1}(f)} \geq C_{n,q}\frac{0.1}{g_{n-2,1}(f)}.
\end{equation}
The bounds in~\eqref{eq:heuristic2} and~\eqref{eq:heuristic3} tell us the term $M_{n,q}$ is $\Theta(f^{-(n-3)})$, which suggests that the number of trees required by heuristic ASTRAL-I is asymptotically significantly larger than that of ASTRAL*. 

Thus our theoretical analysis shows that, consistently with simulation results presented in~\cite{astral2}, for large $n$ there should be conditions where ASTRAL-I does not work well, but ASTRAL* has high accuracy. Note that ASTRAL-II has sought to close this gap heuristically and seems to have succeeded to a large extent based on simulation results~\cite{astral2}.

\end{document}